%% file: MfUSampler.tex
\documentclass[codesnippet]{jss}
\usepackage[latin1]{inputenc}
\usepackage{amsmath}
\usepackage{amsfonts}
\usepackage{amssymb}
\usepackage{graphicx}
\usepackage{color}
\usepackage{amsthm}

\newcommand{\bbeta}{\boldsymbol\beta}
\newcommand{\ggamma}{\boldsymbol\gamma}
\newcommand{\xx}{\mathbf{x}}
\newcommand{\yy}{\mathbf{y}}
\newcommand{\XX}{\mathbf{X}}
\newcommand{\llog}{\mathrm{log}}
\newcommand{\sigmamax}{\sigma_{\mathrm{max}}}
\newcommand{\dd}{\mathrm{d}}
\newtheorem{lemma}{Lemma}

\author{Alireza S. Mahani\\ Scientific Computing Group \\ Sentrana Inc. \And
Mansour T.A. Sharabiani\\ National Heart and Lung Institute \\ Imperial College London}
\title{Multivariate-from-Univariate MCMC Sampler: \proglang{R} Package \pkg{MfUSampler}}

\Plainauthor{Alireza S. Mahani, Mansour T.A. Sharabiani} 
\Plaintitle{Multivariate-from-Univariate MCMC Sampler: R Package MfUSampler} 
\Shorttitle{Multivariate-from-Univariate MCMC Sampler: \proglang{R} Package \pkg{MfUSampler}} 

\Abstract{ 
The \proglang{R} package \pkg{MfUSampler} provides Monte Carlo Markov Chain machinery for generating samples from multivariate probability distributions using univariate sampling algorithms such as Slice Sampler and Adaptive Rejection Sampler. The sampler function performs a full cycle of univariate sampling steps, one coordinate at a time. In each step, the latest sample values obtained for other coordinates are used to form the conditional distributions. The concept is an extension of Gibbs sampling where each step involves, not an independent sample from the conditional distribution, but a Markov transition for which the conditional distribution is invariant. The software relies on proportionality of conditional distributions to the joint distribution to implement a thin wrapper for producing conditionals. Examples illustrate basic usage as well as methods for improving performance. By encapsulating the multivariate-from-univariate logic, \pkg{MfUSampler} provides a reliable library for rapid prototyping of custom Bayesian models while allowing for incremental performance optimizations such as utilization of conjugacy, conditional independence, and porting function evaluations to compiled languages.
}
\Plainkeywords{monte carlo markov chain, slice sampler, adaptive rejection sampler, gibbs sampling} 

\Address{
Alireza S. Mahani\\
Scientific Computing Group\\
Sentrana Inc.\\
1725 I St NW\\
Washington, DC 20006\\
E-mail: \email{alireza.mahani@sentrana.com}\\
}

\usepackage{Sweave}
\begin{document}
\input{MfUSampler-concordance}


\section{Introduction}\label{section-introduction}
Bayesian inference software such as \proglang{Stan}~\citep{stan-software:2014}, \proglang{OpenBUGS}~\citep{thomas2006making}, and \proglang{JAGS}~\citep{plummer2004jags} provide high-level, domain-specific languages (DSLs) to specify and sample from probabilistic Directed Acyclic Graphs (DAGs). In some Bayesian projects, the convenience of using such DSLs comes at the price of reduced flexibility in model specification, and suboptimality of the underlying sampling algorithms used by the compilers. Furthermore, for large projects the end-goal might be to implement all or part of the sampling algorithm in a high-performance - perhaps parallel - language. In such cases, researchers may choose to start their development work by `rolling their own' joint probability distributions from the DAG specification, followed by application of their choice of a sampling algorithm to the joint distribution.

Many Monte Carlo Markov Chain (MCMC) algorithms have been proposed over the years for sampling from complex posterior distributions. Perhaps the most widely-known algorithm is Metropolis~\citep{metropolis1953equation} and its generalization, Metropolis-Hastings (MH)~\citep{hastings1970monte}. These multivariate algorithms are very easy to implement, but they can be slow to converge without a carefully-selected proposal distribution. A particular flavor of MH is the Stochastic Newton Sampler~\citep{qi2002hessian}, where the proposal distribution is a multivariate Gaussian based on the second-order Taylor series expansion of the log-probability. This method has been implemented in the \proglang{R} package \pkg{sns}~\citep{mahani2014sns}. This algorithm can be quite effective for twice-differentiable, log-concave distributions such as those encountered in Generlized Linear Regression (GLM) problems. Hamiltonian Monte Carlo (HMC) algorithms~\citep{girolami2011riemann,neal2011mcmc} have also gained popularity due to development of techniques for automated tuning of their parameters~\citep{hoffman2014no}.

Univariate samplers tend to have few tuning parameters and thus well suited for black-box MCMC software. Two important examples are Adaptive Rejection Sampling~\citep{gilks1992adaptive} (or ARS) and Slice Sampling~\citep{neal2003slice}. ARS requires log-density to be concave, and needs the first derivative, while slice sampler is generic and derivative-free. To apply these univariate samplers to multivariate distributions, they must be applied one-coordinate-at-a-time according to the Gibbs sampling algorithm~\citep{geman1984stochastic}, where at the end of each univariate step the sampled value is used to update the conditional distribution for the next coordinate. \pkg{MfUSampler} encapsulates this logic into a library function, providing a fast and reliable path towards Bayesian model estimation for researchers working on novel DAG specifications.

When posterior distribution exhibits strong correlation structure, one-coordinate-at-a-time algorithms can become inefficient as they fail to capture important geometry of the space~\citep{girolami2011riemann}. This has been a key motivation for research on black-box multivariate samplers, such as adaptations of slice sampler~\citep{thompson2011slice}.

The rest of this article is organized as follows. In Section~\ref{section-implementation} we provide a brief overview of the extended Gibbs sampling framework used in \pkg{MfUSampler}. In Section~\ref{section-usage} we illustrate how to use the software with an example. Section~\ref{section-workflow} shows how \pkg{MfUSampler} can be used to expedite the prototyping stage in Bayesian modeling problems. Finally, Section~\ref{section-summary} provides a summary and concluding remarks.

\section{Implementation}\label{section-implementation}
\pkg{MfUSampler} relies on three components:
\begin{enumerate}
\item Univariate MCMC samplers: As of version 0.9.1 of \pkg{MfUSampler}, two such samplers are supported: Univariate Slice Sampler with Stepout and Shrinkage~\citep{neal2003slice} and Adaptive Rejection Sampler~\citep{gilks1992adaptive}. For slice sampler, we have imported - with small modifications - Radford Neal's \proglang{R} code, posted on his website\footnote{\url{http://www.cs.toronto.edu/~radford/ftp/slice-R-prog}}, while for ARS we use the \proglang{R} package \pkg{ars}~\citep{rodrigues2014ars}. For technical details on the univariate sampling algorithms, see aforementioned publications or statistical textbooks~\citep{robert1999monte}.
\item Main sampling routine, \code{MfU.Sample}: This function is essentially a \code{for} loop for applying the underlying univariate sampler to each coordinate of the multivariate distribution. We refer to this as `extended Gibbs sampling' (Section~\ref{subsection-extended-gibbs}). The function \code{MfU.Control} allows the user to set the tuning parameters of the univariate sampler. Both \code{MfU.Sample} and \code{MfU.Control} are public functions, and their detailed behavior can be examined by consulting the package documentation.
\item Wrapper functions \code{MfU.fEval}, \code{MfU.fgEval.f} and \code{MfU.fgEval.g} (all internal functions) that return the conditional distribution and its gradients for each coordinate, using the underlying joint distribution and its gradient vector (Section~\ref{subsection-conditional-joint}).
\end{enumerate}

\subsection{Extended Gibbs Sampling}\label{subsection-extended-gibbs}
The \code{for} loop inside \code{MfU.Sample} is a direct implementation of Gibbs sampling~\citep{bishop2006pattern}, with one conceptual extension: rather than requiring an independent sample from each coordinate's conditional distribution, we expect a Markov transition for which the conditional distribution is an invariant distribution. Among the current univariate samplers implemented in \pkg{MfUSampler}, Adaptive Rejection Sampler produces a standard Gibbs sampler while the Slice Sampler falls under the `extended' Gibbs sampler. The following lemma provides the proof of invariance. (For a discussion of ergodicity for slice sampler, see ~\cite{roberts1999convergence}).
\begin{lemma}
If a coordinate-wise Markov transition leaves the conditional distribution invariant, it will also leave the joint distribution invariant.
\end{lemma}

\begin{proof}
The premise can be mathematically expressed as
\begin{equation}\label{equation-premise}
p(x'_k | \xx_{\setminus k}) = \int_{x_k} T(x'_k, x_k | \xx_{\setminus k}) p(x_k | \xx_{\setminus k}) \, \dd x_k,
\end{equation}
while the conclusion can be expressed as
\begin{equation}\label{equation-conclusion}
p(x'_k , \xx_{\setminus k}) = \int_{x_k} T(x'_k, x_k | \xx_{\setminus k}) p(x_k , \xx_{\setminus k}) \, \dd x_k.
\end{equation}
In the above $\xx_{\setminus k}$ denotes all coordinates except for $x_k$ and $T(x'_k, x_k | \xx_{\setminus k})$ denotes the coordinate-wise Markov transition density from $x'_k$ to $x_k$. Employing the Product Rule of Probability, we have $p(x_k , \xx_{\setminus k}) = p(x'_k | \xx_{\setminus k}) \times p(\xx_{\setminus k})$. Since the coordinate-wise Markov transition does not change $\xx_{\setminus k}$, we can factor $p(\xx_{\setminus k})$ out of the integral, thereby easily reducing Equation ~\ref{equation-conclusion} to Equation~\ref{equation-premise}.
\end{proof}
Note that standard Gibbs sampling is a special case of the above lemma where $T(x'_k, x_k | \xx_{\setminus k}) = p(x'_k | \xx_{\setminus k})$. (The reader can easily verify that this special transition density satifies the premise.) A full Gibbs cycle is simply a succession of coordinate-wise Markov transitions, and since each one leaves the target distribution invariant according to the above lemma, same is true of the resulting composite Markov transition density.

\subsection{Producing the Conditional Distributions}\label{subsection-conditional-joint}
The internal function \code{MfU.fEval} is responsible for producing coordinate-wise conditional distributions used by the slice sampler:
\begin{Schunk}
\begin{Sinput}
R> MfU.fEval <- function(xk, k, x, f, ...) {
+    x[k] <- xk
+    return (f(x, ...))
+  }
\end{Sinput}
\end{Schunk}
The implementation is deceptively simple, but warrants some explanation. The funcion accepts \code{xk} -  the value of the \code{k}'th coordinate - and inserts it into the $K$-dimensional vector \code{x}. It then evaluates and returns the joint distribution \code{f} at \code{x} (fixed arguments are passed via \code{...}). Returning the joint distribution in lieu of the conditional distribution is correct because, from the perspective of each coordinate, the two are proportional:
\begin{equation}
p(x_k | \xx_{\setminus k}) = \frac{p(x_k , \xx_{\setminus k})}{p(\xx_{\setminus k})} \propto p(x_k , \xx_{\setminus k})
\end{equation}
where the first step follows from the Product Rule and in the second step, we have taken advantage of the fact that $p(\xx_{\setminus k})$ is constant in terms of $x_k$. A multiplicative constant for density translates into an additive constant for log-density, and can be safely ignored in most MCMC algorithms, including the slice sampler and ARS.

Functions \code{MfU.fgEval.f} and \code{MfU.fgEval.g} produce log-density and its gradient for each coordinate, to be consumed by ARS. Since \proglang{R} package \pkg{ars} expects vector forms of input and output for log-density and its gradient, the above two functions implement such vectorization.

\section[Using MfUSampler]{Using \pkg{MfUSampler}}\label{section-usage}
In this section, we illustrate the use of \pkg{MfUSampler} using both slice sampler and ARS as underlying univariate distributions.
\subsection{Example 1: Bayesian Logistic Regression}\label{subsection-usage-example}
We use a Bayesian logistic regression problem ($N$ observations, $K$ coefficients) as an example. The DAG corresponding to this simple problem is shown in Figure~\ref{figure-logistic-dag}. We assume a non-informative, Gaussian prior on each of the $K$ regression coefficients in $\bbeta$ (assuming proper scaling of the covariate matrix $\XX$), using a mean of $\mu=0.0$ and a standard deviation of $\sigma=1e+6$. The full joint distribution corresponding to this DAG is given in Equation~\ref{equation-logprob-logistic}.

\begin{figure}[t]
\centering \includegraphics[width=0.4\textwidth]{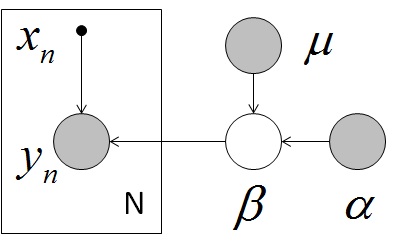} 
\caption{
Directed Acyclic Graph representing a Bayesian logistic regression problem, with $N$ observations and $K$ coefficients. A non-informative Gaussian prior is imposed on each of the $K$ elements of $\bbeta$ using mean $\mu=0.0$ and standard deviation $\sigma=$1e+6. Plate notation is used to show the links from $\bbeta$ to each of the $N$ response values $y_n$.
}
\label{figure-logistic-dag}
\end{figure}

\begin{equation} \label{equation-logprob-logistic}
L(\bbeta) = - \sum_{n=1}^N \left\{ (1-y_n) \xx_n^t \bbeta + \llog(1 + \exp(-\xx_n^t \bbeta)) \right\} - \frac{1}{2 \sigma^2} \sum_{k=1}^K (\beta_k - \mu)^2 + C,
\end{equation}
where the first term corresponds to the log-likelihood, the second term represents the log-prior, and $C$ captures any terms that are independent of $\bbeta$. Note that, while the log-prior term is separable in $\bbeta$, same is not true for log-likelihood. This means the joint log-distribution $L(\bbeta)$ cannot be written as $\sum_k L_k(\beta_k)$. Therefore, conditional distributions used in each of the $K$ steps in a Gibbs sampling cycle will depend on the latest draws from the other $K-1$ elements.

\subsection[]{Slice Sampling using \pkg{MfUSampler}}\label{subsection-usage-slice}
First, we must load the package into an \proglang{R} session. We also set the random seed for reporduciblity of results:
\begin{Schunk}
\begin{Sinput}
R> library("MfUSampler")
R> set.seed(0)
\end{Sinput}
\end{Schunk}
Applying slice sampler to a log-density in \pkg{MfUSampler} is quite straightforward. First we must implement a function that returns the log-density. For the log-density of Equation~\ref{equation-logprob-logistic}, we can do:
\begin{Schunk}
\begin{Sinput}
R> logit.f <- function(beta, X, y, mu=0.0, sigma=1e+6) {
+    Xbeta <- X %*% beta
+    return (-sum((1-y) * Xbeta + log(1 + exp(-Xbeta))) 
+            - sum((beta - mu)^2)/(2*sigma^2))
+  }
\end{Sinput}
\end{Schunk}
where we have ignored the constant term $C$. Note that the first argument to \code{logit.f} is the argument for which we want to draw samples, which is \code{beta} in this case. Also note that \code{X} is a $N$-by-$K$ matrix whose n'th row corresponds to $\xx_n$ in Equation~\ref{equation-logprob-logistic}. To test \pkg{MfUSampler} on \code{logit.f}, we first generate some simulated data:
\begin{Schunk}
\begin{Sinput}
R> N <- 1000
R> K <- 5
R> X <- matrix(runif(N*K, -0.5, +0.5), ncol=K)
R> beta <- runif(K, -0.5, +0.5)
R> y <- 1*(runif(N) < 1/(1+exp(-X %*% beta)))
\end{Sinput}
\end{Schunk}
We now initialize $\bbeta$ with zeros, draw samples from \code{logit.f}, and save the sampled coefficients to \code{beta.smp}:
\begin{Schunk}
\begin{Sinput}
R> nsmp <- 100
R> beta.ini <- rep(0.0, K)
R> beta.smp <- array(NA, dim=c(nsmp,K))
R> for (i in 1:nsmp) {
+    beta.ini <- MfU.Sample(beta.ini, f=logit.f, uni.sampler="slice", X=X, y=y)
+    beta.smp[i,] <- beta.ini
+  }
\end{Sinput}
\end{Schunk}
We can compare the mean of the sampled coefficients (\code{beta.mcmc}) with the MLE estimate (\code{beta.glm}), throwing away the first half of the samples for burn-in:
\begin{Schunk}
\begin{Sinput}
R> beta.mcmc <- colMeans(beta.smp[(nsmp/2+1):nsmp,])
R> beta.glm <- glm(y~X-1, family="binomial")$coefficients
R> cbind(beta.glm,beta.mcmc)
\end{Sinput}
\begin{Soutput}
     beta.glm  beta.mcmc
X1 -0.4265100 -0.4013649
X2 -0.5416799 -0.5814078
X3 -0.1684548 -0.1717733
X4 -0.4056165 -0.4323484
X5  0.4914152  0.5592871
\end{Soutput}
\end{Schunk}
Increasing \code{nsmp} should bring \code{beta.mcmc} closer to \code{beta.glm}.

\subsection[]{Adaptive Rejection Sampling using \pkg{MfUSampler}}\label{subsection-usage-ars}
In order to perform ARS using \pkg{MfUSampler}, we must construct a function that outputs either the log-density or its gradient vector, depending on the value of the boolean argument \code{grad}:
\begin{Schunk}
\begin{Sinput}
R> logit.fg <- function(beta, X, y, mu=0.0, sigma=1e+6, grad) {
+    Xbeta <- X %*% beta
+    if (grad) return (t(X) %*% (1/(1+exp(Xbeta)) - (1-y))
+                      - (beta - mu)/(2*sigma^2))
+    return (logit.f(beta, X, y, mu, sigma))
+  }
\end{Sinput}
\end{Schunk}
We can now apply \code{MfU.Sample}, this time setting \code{uni.sampler} argument to \code{"ars"}:
\begin{Schunk}
\begin{Sinput}
R> beta.ini <- rep(0.0, K)
R> beta.smp <- array(NA, dim=c(nsmp,K))
R> for (i in 1:nsmp) {
+    beta.ini <- MfU.Sample(beta.ini, f=logit.fg, uni.sampler="ars", X=X, y=y)
+    beta.smp[i,] <- beta.ini
+  }
R> beta.mcmc <- colMeans(beta.smp[(nsmp/2+1):nsmp,])
R> cbind(beta.glm,beta.mcmc)
\end{Sinput}
\begin{Soutput}
     beta.glm  beta.mcmc
X1 -0.4265100 -0.4118904
X2 -0.5416799 -0.4960775
X3 -0.1684548 -0.1691510
X4 -0.4056165 -0.4288610
X5  0.4914152  0.4751731
\end{Soutput}
\end{Schunk}
Again, larger values of \code{nsmp} would lead to closer match between \code{beta.mcmc} and \code{beta.glm}.

So far, we ignored the control parameters of Slice sampler and ARS, thus falling back to the default values provided for each algorithm. Sometimes it is necessary to override the default values, which is possible by calling \code{MfU.Control}. To see a description of the control parameters, type \code{?MfU.Control} in your \proglang{R} session. More details on each algorithm and its tuning parameters can be found by consulting the paper~\citep{neal2003slice} for the Slice sampler, and the software documentation~\citep{rodrigues2014ars} as well as the paper~\citep{gilks1992adaptive} for ARS.

The Bayesian logistic regression example is somewhat contrived. For example, the log-posterior of Equation~\ref{equation-logprob-logistic} can be shown to have a negative-definite Hessian matrix and therefore log-concave, thus being eligible for the multivariate Stochastic Newton Sampling~\citep{mahani2014sns}. The real power of \pkg{MfUSampler}, however, comes from its application to novel Bayesian models where developers seek to rapidly prototype flexible DAGs to validate their model specification, while having a path towards high-peformance applications. Section~\ref{section-workflow} contains one such example.

\section[]{\pkg{MfUSampler} for Bayesian Prototyping}\label{section-workflow} 
With \pkg{MfUSampler}, prototyping a novel Bayesian model can be quite fast. Encapsulation of Gibbs cycles inside the \code{MfU.Sample} function circumvents subtle bugs during implementation of Gibbs sampling, and allows the researcher to focus his/her mental power on other, more important complexities of the problem such as model specification, MCMC diagnostics, etc.

\subsection{Example 2: Heteroscedastic Linear Regression}\label{subsection-example-heteroscedastic}
Consider a linear regression problem where we suspect heteroscedasticity of regression residuals, and that the non-uniform residual variance is itself dependent on the same covariates that are used to explain the mean response. To model such behavior, we assume the $N$ response measurements, $y_n$,  are independently drawn from this normal distribution:
\begin{equation} \label{equation-regression-heteroscedastic}
y_n \sim \mathcal{N}(\xx_n^t \bbeta, \sigmamax^2/(1+\exp(-\xx_n^t \ggamma)))
\end{equation}
This specification digresses from ordinary linear regression by making the variance parameter data-dependent. The nonlinear transformation ensures that the variance for all points is bounded between $0$ and $\sigmamax^2$. Given the matrix of covariates $\XX$ and the response vector $\yy$, our goal is to estimate $\bbeta$, $\ggamma$ and $\sigmamax$. There can be several reasons for our interest in a Bayesian treatment of this problem, and in drawing samples from the joint posterior distribution on $(\bbeta, \ggamma, \sigmamax)$: 1) future extension of the model in a probabilistic framework, e.g. using a Hierrarchical Bayesian approach to pool data across heterogneous observation units, 2) our suspicion that the posterior distribution is not globally convex, has a complex shape with multiple possible maxima, and thus sampling the its entire landscape can be safer than looking for a single, local optimum, and 3) fully probabilistic treatment of parameter estimation as well as response prediction.

To avoid clutter, we assume non-informative priors on all parameters and ignore them to focus only on the log-likelihood function corresponding to Equation~\ref{equation-regression-heteroscedastic}, which can be implemented using this \proglang{R} function:
\begin{Schunk}
\begin{Sinput}
R> loglike <- function(beta, gamma, sigmamax, X, y) {
+    mean.vec <- X%*%beta
+    sd.vec <- sigmamax/sqrt(1+exp(-X%*%gamma))
+    return (sum(dnorm(y, mean.vec, sd.vec, log=TRUE)))
+  }
\end{Sinput}
\end{Schunk}
To feed the log-likelihood function into \code{MfU.Sample}, we write a thin wrapper around it:
\begin{Schunk}
\begin{Sinput}
R> loglike.wrapper <- function(coeff, X, y) {
+    K <- ncol(X)
+    beta <- coeff[1:K]
+    gamma <- coeff[K+1:K]
+    sigmamax <- coeff[2*K+1]
+    return (loglike(beta, gamma, sigmamax, X, y))
+  }
\end{Sinput}
\end{Schunk}
We can now generate some data and draw samples from the resulting log-likelihood. Note the use of \code{MfU.Control} to set a lower bound of \code{1e-3} on \code{sigmax}:
\begin{Schunk}
\begin{Sinput}
R> # generate simulated data
R> K <- 5
R> N <- 1000
R> X <- matrix(runif(N*K, -0.5, +0.5), ncol=K)
R> beta <- runif(K, -0.5, +0.5)
R> gamma <- runif(K, -0.5, +0.5)
R> sigmamax <- 0.75
R> mu <- X%*%beta
R> var <- sigmamax^2/(1+exp(-X%*%gamma))
R> y <- rnorm(N, mu, sqrt(var))
R> # initialize and sample
R> coeff <- c(rep(0.0, 2*K), 0.5)
R> mycontrol <- MfU.Control(n = 2*K+1, slice.lower = c(rep(-Inf,2*K), 0.001))
R> coeff.smp <- array(NA, dim=c(nsmp, 2*K+1))
R> t <- proc.time()[3]
R> for (i in 1:nsmp) {
+    coeff <- MfU.Sample(coeff, f=loglike.wrapper, X=X, y=y, control = mycontrol)
+    coeff.smp[i,] <- coeff
+  }
R> t <- proc.time()[3]-t
R> cat("time:", t, "\n")
\end{Sinput}
\begin{Soutput}
time: 1.718 
\end{Soutput}
\begin{Sinput}
R> beta.est <- colMeans(coeff.smp[(nsmp/2+1):nsmp, 1:K])
R> gamma.est <- colMeans(coeff.smp[(nsmp/2+1):nsmp, K+1:K])
R> sigmamax.est <- mean(coeff.smp[(nsmp/2+1):nsmp, 2*K+1])
R> cbind(beta, beta.est, gamma, gamma.est)
\end{Sinput}
\begin{Soutput}
            beta    beta.est     gamma  gamma.est
[1,] -0.08169418  0.01082307 0.3845259  0.3442652
[2,] -0.36962596 -0.33175088 0.4858229  0.4563613
[3,]  0.46721080  0.42836681 0.2954812  0.1705552
[4,] -0.17247123 -0.25379361 0.3079317  0.3556921
[5,]  0.33722306  0.23397001 0.1617440 -0.3678595
\end{Soutput}
\begin{Sinput}
R> c(sigmamax, sigmamax.est)
\end{Sinput}
\begin{Soutput}
[1] 0.750000 0.753041
\end{Soutput}
\end{Schunk}
Note that we applied \code{MfU.Sample} to \code{sigmamax} even though it is a scalar rather than a vector. In principle, we could directly call the slice sampler, but using the same higher level function \code{MfU.Sample} call keeps the code more organized and easier to track. Despite \code{nsmp} being small, we see general agreement between actual and estimated parameters, especially for \code{sigsq} and \code{beta}. Of course, proper MCMC diagnostics including trace plot examination, histogram examination, and effective size calculation must be done to adjust sampling parameters and determine the next step in the modeling process, including model re-specification.

\subsection[]{Performance Improvement Techniques}\label{subsection-workflow-efficiency}
Feeding the full joint distribution for a DAG into \code{MfU.Sample} is often a good, first step but it can be computationally brute-force. There are several opportunities for performance improvement once the foundation is laid and the model structure is somewhat validated. The root-cause of inefficiency in our brute-force approach is that we are evaluating the full, joint density during univariate sampling of each coordinate of the state space. This may not be necessary for several reasons:
\begin{enumerate}
\item For some variables or variable groups, exact sampling techniques may be possible. This often arises when the likelihood and prior terms in a hierarchical model are conjugate, making the posterior distribution have the same functional form as the likelihood. When conjugacy allows for exact sampling, it is often the preferred route compared to MCMC sampling~\citep{robert1999monte}.
\item If some variables (or groups of variables) are conditionally-independent, their joint distribution, conditioned on the remaining variables, is separable~\citep{wilkinson2006parallel}. Therefore, while sampling each variable, we only need to evaluate a subset of additive terms comprising log-likelihood. This situation can arise, for example, in Hierarchical Bayesian regression models~\citep{rossi2003bayesian,gelman2006data}. In addition to permitting lighter conditional posterior evaluations, conditional independence can also be taken advantage of in parallel Gibbs sampling~\citep{wilkinson2006parallel}.
\item Even when conditional independence does not exist, we may still find opportunities to drop some of the additive terms in the log-posterior for all or a subset of the variables, thereby reducing the time needed to evaluate conditional posteriors during Gibbs cycles.
\end{enumerate}

In the above cases, the general strategy is to split the full state space into subspaces and apply the more efficient techniques within each subspace. Bayesian compilers are effective to varying degrees at identifying and taking advantage of such optimization opportunities. A detailed discussion of these topics is beyond the scope of this paper. Here we illustrate the last item in the above list, continuing the heteroscedastic linear regression example of Section~\ref{subsection-example-heteroscedastic}.

The \code{loglike} function of Section~\ref{subsection-example-heteroscedastic} can be expanded as below (ignoring constant terms):
\begin{eqnarray}\label{equation-heteroscedastic-terms}\nonumber
L(\bbeta, \ggamma, \sigmamax) &=& -N \llog \sigmamax + \frac{1}{2} \sum_{n=1}^N \llog \left( 1 + \exp(- \xx_n^t \ggamma) \right) \\
&& - \frac{1}{2 \sigmamax^2} \sum_{n=1}^N \left\{ (y_n - \xx_n^t \bbeta)^2 \left( 1 + \exp (- \xx_n^t \ggamma) \right) \right\}
\end{eqnarray}
We see that, of the above three additive terms, only the last term depends on all three variable blocks $\bbeta$, $\ggamma$ and $\sigmamax$, while the first two depend only on $\sigmamax$ and $\ggamma$, respectively. Taking advantage of this, we can write simplified conditional log-likelihood functions as below:
\begin{align}
\left\{
\begin{array}{lll}
L(\bbeta \,|\, \ggamma, \sigmamax) &= - \frac{1}{2 \sigmamax^2} \sum_{n=1}^N \left\{ (y_n - \xx_n^t \bbeta)^2 \left( 1 + \exp (- \xx_n^t \ggamma) \right) \right\} \\ \\
L(\ggamma \,|\, \bbeta, \sigmamax) &=  \frac{1}{2} \sum_{n=1}^N \llog \left( 1 + \exp(- \xx_n^t \ggamma) \right) - \frac{1}{2 \sigmamax^2} \sum_{n=1}^N \left\{ (y_n - \xx_n^t \bbeta)^2 \left( 1 + \exp (- \xx_n^t \ggamma) \right) \right\} \\ \\
L(\sigmamax \,|\, \bbeta, \ggamma) &= -N \llog \sigmamax - \frac{1}{2 \sigmamax^2} \sum_{n=1}^N \left\{ (y_n - \xx_n^t \bbeta)^2 \left( 1 + \exp (- \xx_n^t \ggamma) \right) \right\}
\end{array}
\right.
\end{align}
Referring to the three terms on the right-hand side of Equation~\ref{equation-heteroscedastic-terms} as components 1-3, the \proglang{R} implementation of the above conditional distributions will be:
\begin{Schunk}
\begin{Sinput}
R> loglike.component1 <- function(sigmamax, N) -N*log(sigmamax)
R> loglike.component2 <- function(gamma, X) 0.5*sum(log(1 + exp(-X%*%gamma)))
R> loglike.component3 <- function(beta, gamma, sigmamax, X, y) {
+    -sum((y-X%*%beta)^2*(1+exp(-X%*%gamma)))/(2*sigmamax^2)
+  }
R> loglike.beta <- function(beta, gamma, sigmamax, X, y) {
+    loglike.component3(beta, gamma, sigmamax, X, y)
+  }
R> loglike.gamma <- function(gamma, beta, sigmamax, X, y) {
+    loglike.component2(gamma, X) +
+      loglike.component3(beta, gamma, sigmamax, X, y)
+  }
R> loglike.sigmamax <- function(sigmamax, beta, gamma, sigma, X, y) {
+    loglike.component1(sigmamax, nrow(X)) +
+      loglike.component3(beta, gamma, sigmamax, X, y)
+  }
\end{Sinput}
\end{Schunk}
Each Gibbs cycle will be broken into 3 steps, corresponding to $\bbeta$, $\ggamma$ and $\sigmamax$:
\begin{Schunk}
\begin{Sinput}
R> beta.ini <- rep(0.0, K)
R> gamma.ini <- rep(0.0, K)
R> sigmamax.ini <- 0.5
R> mycontrol.sigmamax <- MfU.Control(n = 1, slice.lower = 0.001)
R> coeff.smp <- array(NA, dim=c(nsmp, 2*K+1))
R> t <- proc.time()[3]
R> for (i in 1:nsmp) {
+    beta.ini <- MfU.Sample(beta.ini, loglike.beta, gamma=gamma.ini
+                   , sigmamax=sigmamax.ini, X=X, y=y)
+    gamma.ini <- MfU.Sample(gamma, loglike.gamma, beta=beta.ini
+                   , sigmamax=sigmamax.ini, X=X, y=y)
+    sigmamax.ini <- MfU.Sample(sigmamax, loglike.sigmamax
+                   , beta=beta.ini, gamma=gamma.ini
+                   , X=X, y=y, control = mycontrol.sigmamax)
+    coeff.smp[i,] <- c(beta.ini, gamma.ini, sigmamax.ini)
+  }
R> t <- proc.time()[3]-t
R> cat("time:", t, "\n")
\end{Sinput}
\begin{Soutput}
time: 1.589 
\end{Soutput}
\begin{Sinput}
R> beta.est <- colMeans(coeff.smp[(nsmp/2+1):nsmp, 1:K])
R> gamma.est <- colMeans(coeff.smp[(nsmp/2+1):nsmp, K+1:K])
R> sigmamax.est <- mean(coeff.smp[(nsmp/2+1):nsmp, 2*K+1])
R> cbind(beta, beta.est, gamma, gamma.est)
\end{Sinput}
\begin{Soutput}
            beta      beta.est     gamma  gamma.est
[1,] -0.08169418  0.0009252645 0.3845259  0.3583743
[2,] -0.36962596 -0.3330523615 0.4858229  0.4546368
[3,]  0.46721080  0.4348825919 0.2954812  0.0726346
[4,] -0.17247123 -0.2702465415 0.3079317  0.2755692
[5,]  0.33722306  0.2422290123 0.1617440 -0.3499572
\end{Soutput}
\begin{Sinput}
R> c(sigmamax, sigmamax.est)
\end{Sinput}
\begin{Soutput}
[1] 0.7500000 0.7532132
\end{Soutput}
\end{Schunk}
In this case, the time savings from our improvements is modest, but for other problems the impact can be more pronounced. For example, for HB regression problems the speedup from breaking down the conditional posterior across regression groups will roughly be proportional to the number of groups, even before applying any parallelization (which could theoretically offer another multiplicative speedup equal to number of groups, given sufficient number of processing cores available).

For statistical problems of moderate to large size (i.e. number of observations and/or covariates) and in the absence of conjugacy for coefficients that are directly involved in explaining the response, the majority of time is often spent in log-density evaluations, rather than other activities such as random number generation or the sampling algorithm itself. Therefore, the next most rewarding optimization step is likely to be porting of log-density functions to high-performance languages such as \proglang{C}, \proglang{C++}, \proglang{FORTRAN}. For large problems, even parallel hardware such as Graphic Processing Units (GPUs) can be utilized by writing log-density functions in languages such as \proglang{CUDA}, while continuing to take advantage of \pkg{MfUSampler} for sampler control logic. Minimizing data movement between processor and co-processor is a key performance factor in such cases. Finally, should further performance improvement necessitate a rewrite of the \pkg{MfUSampler} logic in a high-peformance language, the package source code can be used as a blue-print for efficient development.

\section{Summary}\label{section-summary}
The \proglang{R} package \pkg{MfUSampler} enables MCMC sampling of multivariate distributions using univariate algorithms. It relies on an extension of Gibbs sampling from univariate independent samplig to univariate Markov transitions, and proportionality of conditional and joint distributions. By encapsulating these two concepts in a library, it reduces the possibility of subtle mistakes by researchers while re-implementing the Gibbs sampler and thus allows them to focus on other, more innovative aspects of their Bayesian modeling. Brute-force application of \pkg{MfUSampler} allows researchers to get their project off the ground, maintain full control over model specification, and utilize robust univariate samplers. This can be followed by an incremental optimization approach by taking advantage of DAG properties such as conjugacy, conditional independence and by porting log-density functions to high-peformance languages and hardware.

\bibliography{MfUSampler}
\end{document}

%% file: MfUSampler-concordance.tex
\Sconcordance{concordance:MfUSampler.tex:MfUSampler.Rnw:%
1 84 1 1 0 4 1 1 4 42 1 1 5 7 0 1 2 27 1 1 2 1 0 1 1 3 0 1 2 1 %
6 8 0 2 2 1 0 4 1 3 0 2 2 1 0 2 1 1 4 6 0 2 2 1 0 2 1 11 0 1 2 %
3 1 1 7 9 0 2 2 1 0 1 1 1 4 3 0 2 1 11 0 1 2 16 1 1 6 8 0 1 2 %
1 8 10 0 1 2 1 3 2 0 8 1 1 2 1 0 3 1 1 4 3 0 2 1 5 0 4 1 10 0 %
1 1 6 0 1 2 27 1 1 2 1 0 1 1 1 3 2 0 1 3 2 0 1 4 3 0 1 4 6 0 2 %
2 1 0 5 1 1 10 9 0 2 1 5 0 4 1 10 0 1 1 6 0 1 2 8 1}

%% file: MfUSampler.bbl
\begin{thebibliography}{21}
\newcommand{\enquote}[1]{``#1''}
\providecommand{\natexlab}[1]{#1}
\providecommand{\url}[1]{\texttt{#1}}
\providecommand{\urlprefix}{URL }
\expandafter\ifx\csname urlstyle\endcsname\relax
  \providecommand{\doi}[1]{doi:\discretionary{}{}{}#1}\else
  \providecommand{\doi}{doi:\discretionary{}{}{}\begingroup
  \urlstyle{rm}\Url}\fi
\providecommand{\eprint}[2][]{\url{#2}}

\bibitem[{Bishop(2006)}]{bishop2006pattern}
Bishop CM (2006).
\newblock \emph{Pattern recognition and machine learning}, volume~1.
\newblock springer New York.

\bibitem[{Gelman and Hill(2006)}]{gelman2006data}
Gelman A, Hill J (2006).
\newblock \emph{Data analysis using regression and multilevel/hierarchical
  models}.
\newblock Cambridge University Press.

\bibitem[{Geman and Geman(1984)}]{geman1984stochastic}
Geman S, Geman D (1984).
\newblock \enquote{Stochastic relaxation, Gibbs distributions, and the Bayesian
  restoration of images.}
\newblock \emph{Pattern Analysis and Machine Intelligence, IEEE Transactions
  on}, (6), 721--741.

\bibitem[{Gilks and Wild(1992)}]{gilks1992adaptive}
Gilks WR, Wild P (1992).
\newblock \enquote{Adaptive rejection sampling for Gibbs sampling.}
\newblock \emph{Applied Statistics}, pp. 337--348.

\bibitem[{Girolami and Calderhead(2011)}]{girolami2011riemann}
Girolami M, Calderhead B (2011).
\newblock \enquote{Riemann manifold langevin and hamiltonian monte carlo
  methods.}
\newblock \emph{Journal of the Royal Statistical Society: Series B (Statistical
  Methodology)}, \textbf{73}(2), 123--214.

\bibitem[{Hastings(1970)}]{hastings1970monte}
Hastings WK (1970).
\newblock \enquote{Monte Carlo sampling methods using Markov chains and their
  applications.}
\newblock \emph{Biometrika}, \textbf{57}(1), 97--109.

\bibitem[{Hoffman and Gelman(2014)}]{hoffman2014no}
Hoffman MD, Gelman A (2014).
\newblock \enquote{The No-U-Turn Sampler: Adaptively Setting Path Lengths in
  Hamiltonian Monte Carlo.}
\newblock \emph{Journal of Machine Learning Research}, \textbf{15}, 1593--1623.

\bibitem[{Mahani \emph{et~al.}(2014)Mahani, Hasan, Jiang, and
  Sharabiani}]{mahani2014sns}
Mahani AS, Hasan A, Jiang M, Sharabiani MT (2014).
\newblock \emph{sns: Stochastic Newton Sampler (SNS)}.
\newblock R package version 0.9.1,
  \urlprefix\url{http://CRAN.R-project.org/package=sns}.

\bibitem[{Metropolis \emph{et~al.}(1953)Metropolis, Rosenbluth, Rosenbluth,
  Teller, and Teller}]{metropolis1953equation}
Metropolis N, Rosenbluth AW, Rosenbluth MN, Teller AH, Teller E (1953).
\newblock \enquote{Equation of state calculations by fast computing machines.}
\newblock \emph{The journal of chemical physics}, \textbf{21}(6), 1087--1092.

\bibitem[{Neal(2011)}]{neal2011mcmc}
Neal R (2011).
\newblock \enquote{MCMC using Hamiltonian dynamics.}
\newblock \emph{Handbook of Markov Chain Monte Carlo}, \textbf{2}.

\bibitem[{Neal(2003)}]{neal2003slice}
Neal RM (2003).
\newblock \enquote{Slice sampling.}
\newblock \emph{Annals of statistics}, pp. 705--741.

\bibitem[{Plummer(2004)}]{plummer2004jags}
Plummer M (2004).
\newblock \enquote{JAGS: Just another Gibbs sampler.}

\bibitem[{Qi and Minka(2002)}]{qi2002hessian}
Qi Y, Minka TP (2002).
\newblock \enquote{Hessian-based markov chain monte-carlo algorithms.}

\bibitem[{Robert and Casella(1999)}]{robert1999monte}
Robert CP, Casella G (1999).
\newblock \emph{Monte Carlo statistical methods}.
\newblock Springer.

\bibitem[{Roberts and Rosenthal(1999)}]{roberts1999convergence}
Roberts GO, Rosenthal JS (1999).
\newblock \enquote{Convergence of slice sampler Markov chains.}
\newblock \emph{Journal of the Royal Statistical Society: Series B (Statistical
  Methodology)}, \textbf{61}(3), 643--660.

\bibitem[{Rodriguez \emph{et~al.}(2014)Rodriguez, Wild, and
  Gilks}]{rodrigues2014ars}
Rodriguez PP, Wild P, Gilks WR (2014).
\newblock \emph{ars: Adaptive Rejection Sampling}.
\newblock R package version 0.5,
  \urlprefix\url{http://CRAN.R-project.org/package=ars}.

\bibitem[{Rossi and Allenby(2003)}]{rossi2003bayesian}
Rossi PE, Allenby GM (2003).
\newblock \enquote{Bayesian statistics and marketing.}
\newblock \emph{Marketing Science}, \textbf{22}(3), 304--328.

\bibitem[{{Stan Development Team}(2014)}]{stan-software:2014}
{Stan Development Team} (2014).
\newblock \enquote{Stan: A C++ Library for Probability and Sampling, Version
  2.5.0.}
\newblock \urlprefix\url{http://mc-stan.org/}.

\bibitem[{Thomas \emph{et~al.}(2006)Thomas, O'Hara, Ligges, and
  Sturtz}]{thomas2006making}
Thomas A, O'Hara B, Ligges U, Sturtz S (2006).
\newblock \enquote{Making BUGS open.}
\newblock \emph{R news}, \textbf{6}(1), 12--17.

\bibitem[{Thompson(2011)}]{thompson2011slice}
Thompson MB (2011).
\newblock \emph{Slice Sampling with Multivariate Steps}.
\newblock Ph.D. thesis, University of Toronto.

\bibitem[{Wilkinson(2006)}]{wilkinson2006parallel}
Wilkinson DJ (2006).
\newblock \enquote{Parallel bayesian computation.}
\newblock \emph{Statistics Textbooks and Monographs}, \textbf{184}, 477.

\end{thebibliography}
